\documentclass[conference]{IEEEtran}
\IEEEoverridecommandlockouts
\usepackage{cite}
\usepackage{amsmath,amssymb,amsfonts}
\usepackage{algorithmic}
\usepackage{graphicx}
\usepackage{textcomp}

\usepackage{amsthm}

\def\BibTeX{{\rm B\kern-.05em{\sc i\kern-.025em b}\kern-.08em
    T\kern-.1667em\lower.7ex\hbox{E}\kern-.125emX}}

\newtheorem{remark}{Remark}
\newtheorem{theorem}{Theorem}
\newtheorem{lemma}{Lemma}
\newtheorem{definition}{Definition}
\newtheorem{corollary}{Corollary}

\makeatletter
\renewcommand*\env@matrix[1][c]{\hskip -\arraycolsep
	\let\@ifnextchar\new@ifnextchar
	\array{*\c@MaxMatrixCols #1}}
\makeatother

\begin{document}

\title{Second-Order Agents on Ring Digraphs
\thanks{The results were obtained with support of the Russian Science Foundation, project no. 16-11-00063 granted to IRE RAS.}
}

\author{
\IEEEauthorblockN{Sergei Parsegov}
\IEEEauthorblockA{\textit{Skolkovo Institute of Science and Technology;} \\
	\textit{V.A. Trapeznikov Institute of Control Sciences} \\
	\textit{of Russian Academy of Sciences}\\
	Moscow, Russia \\
	s.parsegov@skoltech.ru}
\and
\IEEEauthorblockN{Pavel Chebotarev}
\IEEEauthorblockA{\textit{ V.A. Trapeznikov Institute of Control Sciences} \\
	\textit{of Russian Academy of Sciences;}\\
		\textit{V.A. Kotelnikov Institute of Radioengineering and Electronics} \\
	\textit{of Russian Academy of Sciences}\\
	Moscow, Russia \\
	pavel4e@gmail.com}

}

\maketitle

\begin{abstract}
The paper addresses the problem of consensus seeking among second-order linear agents interconnected in a specific ring topology. Unlike the existing results in the field dealing with one-directional digraphs arising in various cyclic pursuit algorithms or two-directional graphs, we focus on the case where some arcs in a two-directional ring graph are dropped in a regular fashion. The derived condition for achieving consensus turns out to be independent of the number of agents in a network.
\end{abstract}


\section{Introduction}
Simple averaging control laws based on local interactions have paved the way to a new class of models in modern control theory and, more widely, in the interconnected dynamical systems theory. Such systems consist of a large (as usual) number of identical subsystems and are supposed to achieve certain global goals. The subsystems, or agents, are coupled in some way and therefore share an amount of common information. During the last 15 years, the complexity of models of networked systems increased significantly starting from a simple continuous-time consensus model (suprisingly first proposed by sociologists in 1964 and rediscovered much later), see \cite{ProskurnikovTempo2017} for details. It is convenient to analyze the evolution of these models considering separately three main entities comprising a dynamical network: agent complexity, interaction structure complexity, and link complexity. Recent results and challeging problems in this field may be found in the lecture course \cite{Bullo2018} and monographs \cite{Lewisetal2014} and \cite{LiDuan2017}. In this paper, we focus on specific interaction structures and analyze their properties in the case of second-order agents.

Control problems related to networks with specific communication patterns play an important role in cooperative/decentralized control. Within this trend, it is supposed that each agent interacts with a predefined number of its indexed neighbors. A pioneering work on consensus where a simple averaging rule was proposed and thoroughly studied was published in 1878 by J.G.~Darboux \cite{Darboux1878}. Although this problem dealt with the evolution of planar polygons, it turned out to be the first theoretical result on discrete-time \emph{cyclic pursuit}. This class of algorithms has a long history (see, e.g., \cite{Nahin2007}, \cite{Sharmaetal2013}, \cite{ElorBruckstein2011}, \cite{MarshallBrouckeFrancis2004}, and references therein), and has a wide range of applications including but not limited to numerous formation control tasks as patrolling, boundary mapping, etc. The cyclic pursuit is a strategy where agent $i$ pursues its neighbor $i - 1$ , while agent 1 pursues agent $n$, thus the topology of communication is a Hamiltonian cycle.  The extensions to hierarchical structures are considered in \cite{SmithBrouckeFrancis2005} and \cite{MukherjeeGhose2016};  \cite{SinhaGhose2006} addresses the case of heterogeneous agents; geometrical problems related to cyclic pursuit-like algorithms are investigated in \cite{ElmachtoubVanLoan2010} and \cite{Shcherbakov2011}. Some pursuit algorithms utilize the rotation operator in order to follow the desired trajectories, see \cite{RamirezRiberosPavoneetal2010} and references therein. Another group of strategies/protocols are based on two-directional topologies, i.e., each agent $i$ has the relative information of the neighbors $i-1$ and $i+1$ (with $0 \equiv n$).  For example, in \cite{MukherjeeZelazo2018} the agents are interconnected by a two-directional ring. The row straightening problems studied in \cite{WagnerBruckstein1997}, \cite{KvintoParsegov2012}, \cite{ProskurnikovParsegov2016} also imply symmetric communications except for the fixed ``anchors'' (the endpoints of the segment). Theoretical motivation behind studying regular network structures is that for some cases, this may lead to the closed-form computation of the spectra of the corresponding Laplacian matrices.

In what follows, we study the problem of reaching consensus for second-order agents with velocity damping (friction). The models of such kind naturally arise,  e.g., in energy systems \cite{Goldin2013} and formation control \cite{RenCao2011, Bullo2018, Ren2008a}. The communication topology  studied in the present paper is a digraph $\mathcal{G}$ with a specific structure: it has $n = 2m$ vertices, $m \geq 3$, and contains a Hamiltonian cycle supplemented by the inverse cycle, where every second arc is dropped, see Fig.~\ref{fig:graph}. In some sense, this ``intermediate'' digraph with regular structure lies between the two-directional ring and the Hamiltonian cycle appearing in cyclic pursuit algorithms. The obtained results fill the gap between undirected and directed ring topologies for this special case: the loss of arcs may result in instability of the whole system and therefore a challenging problem is to derive a suitable condition on the tunable damping parameter that guarantees convergence to consensus irrespective of the number of agents in the network. Alternatively, one may consider the opposite transformation, that is addition of arcs to an unidirected communication structure in order to reduce relatively high damping coefficients of agents.

The paper is organized as follows. The next section presents the notations used in the paper. Section~3 introduces some mathematical preliminaries needed for further exposition and discusses the statement of the problem. The main theoretical result is formulated and proved in Section~4. Finally, the results of numerical simulation and conclusions are given.

\section{Notations}
The following notation will be used throughout the paper:
\begin{itemize}
	\item $j := \sqrt {-1}$ denotes the imaginary unit, the letters $i$ and $k$ are reserved for indices;
 	\item the unit vector is defined by $\textbf{1}_{n}:=[1,1,\ldots,1]^{\top} \in \mathbb{R}^{n}$;
 	\item $\otimes$ denotes the Kronecker product of two matrices.
\end{itemize}

Below we consider some definitions and auxiliary lemmas needed for further discussion and give an exact problem formulation.

\section{Preliminaries and Problem Statement}

Throughout the paper, we consider a group of $n$ identical agents with a directed communication topology and suppose that $n = 2m$, where $m \geq 3$. Each agent obeys the second-order dynamics of the form
\begin{equation}
\label{eq:agent0}
\left\{
  \begin{array}{lr}
    \dot x_i = v_i, \\
    \dot v_i = -\gamma v_i + u_i, 
  \end{array} \quad i \in 1:n,
\right.
\end{equation}
where $x_i, v_i, u_i \in \mathbb{R}$ are the position coordinate of the $i$th agent, its velocity and control input, respectively; $\gamma > 0$ denotes the damping coefficient.

The state-space representation of system \eqref{eq:agent0} writes as
\begin{equation}
\label{eq:agent1}
\dot \xi_i = A\xi_i+Bu_i,
\end{equation}
where $\xi_i=\left[x_i, v_i \right]^{\top}$, $A = \begin{bmatrix}[r] 0 & 1\\ 0 & -\gamma \end{bmatrix}$, and $B = \left[0, 1 \right]^{\top}$.

Suppose that the agents interact through the topology depicted in Fig.~\ref{fig:graph}.
\begin{figure}[h!]
  \centering
    \includegraphics[width=5.5cm]{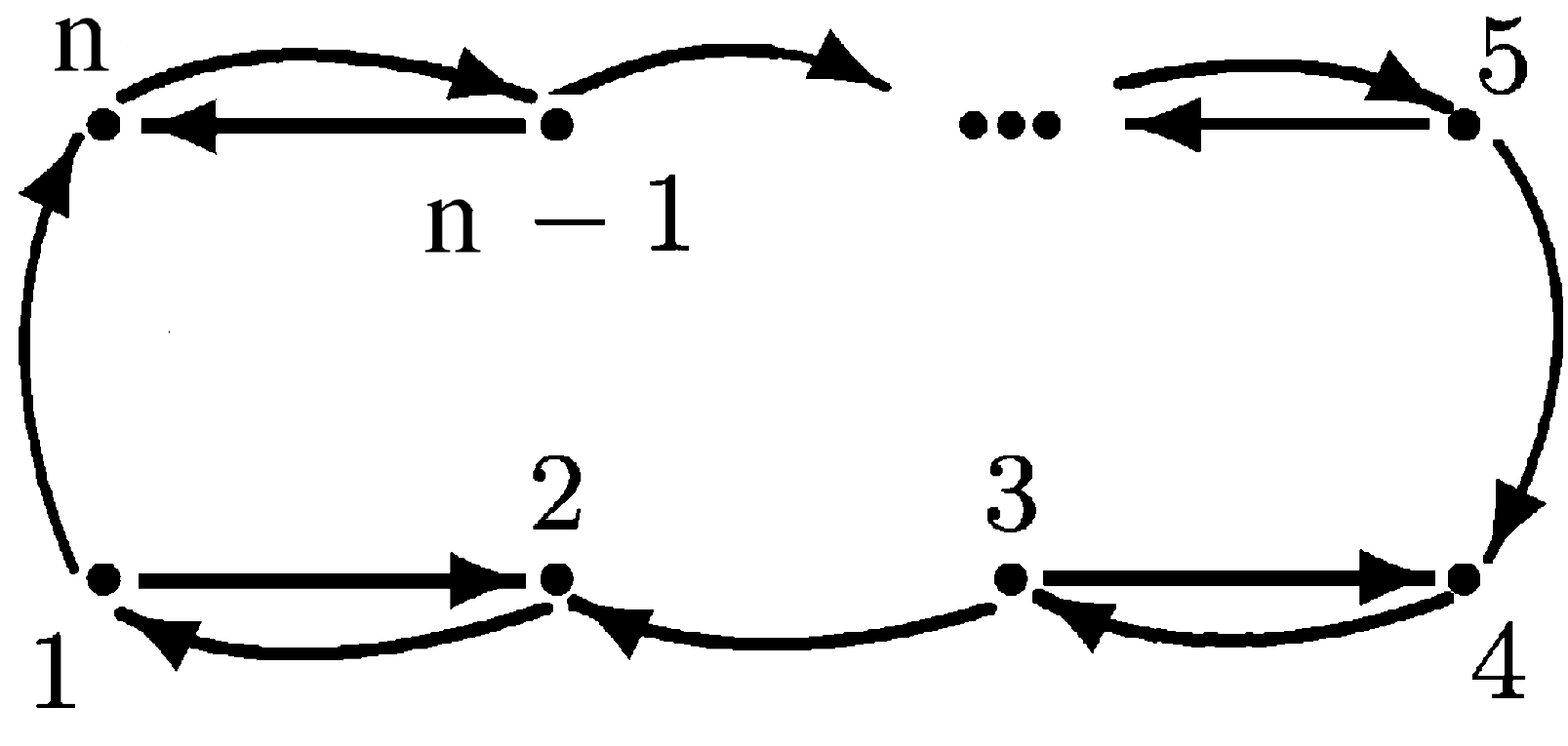}
    \caption{The interaction topology between agents \eqref{eq:agent0}.}
  \label{fig:graph}
\end{figure}
Thus, each agent knows the relative distances between itself and the nearest one or two indexed neighbors. Let $K \in \mathbb{R}^{1 \times 2}$ be the matrix $K = \left[1, 0 \right]$, therefore the decentralized control protocol has the form
\begin{equation}
\label{eq:controlprotocol1}
u_i  \hspace{-2pt}=\hspace{-2pt} \left\{
  \begin{array}{lr}
    \hspace{-6pt}K(\xi_{i+1}-\xi_i) + K(\xi_{i-1}-\xi_i), ~ \text{$i$ is odd}, \\
    \hspace{-6pt}K(\xi_{i-1}-\xi_i), \quad \quad \quad \quad \quad \quad  ~~ \text{$i$ is even}, 
  \end{array}\hspace{-2pt} i \in\hspace{-2pt} 1:n.
\right.
\end{equation}

The main objective of the paper is to establish conditions for a network of second-order agents \eqref{eq:agent0}, \eqref{eq:agent1} governed by protocol \eqref{eq:controlprotocol1} that guarantee consensus in the sense of
\begin{equation}
\label{eq:consensus}
\lim_{t \rightarrow \infty}{\| \xi_i(t) -\xi_k(t) \|} = 0, \; \forall i,k \in 1:n.
\end{equation}
These conditions turn out to be independent on the number of agents comprising the network system.

The following definitions and results are used in the further considerations.

\subsection{Graph Theory}
We suppose that the communication network is represented by a fixed, directed graph $\mathcal{G} = (\mathcal{V,E})$, where $\mathcal{V} = \{1,\ldots,n\}$ denotes the vertex (or node) set and $\mathcal{E}$ is the set of arcs.
We also assume that the graph is simple, i.e., it has no self-loops, and no multiple arcs. Let us define the Laplacian matrix of an unweighted digraph $\mathcal{G}$ as follows:
\begin{definition}[\hspace{-4pt} \cite{AgaevChebotarev2002}]
\label{def:laplacian}
The Laplacian matrix of a digraph $\mathcal{G}$ is the matrix $\mathcal{L} = (l_{ik} ) \in \mathbb{R}^{n \times n} $ in which, for $k \neq i$, $l_{ik} = -1$ whenever $(i,k) \in \mathcal{E(G)}$, otherwise $l_{ik} = 0$. The diagonal entries of $\mathcal{L}$ are of the form $l_{ii} = -\sum_{k \neq i}l_{ik}$, $i,k \in  \mathcal{V(G)}$.
\end{definition}

\begin{lemma}[\hspace{-3pt} \cite{AgaevChebotarev2000}, \cite{AgaevChebotarev2001}]
\label{lem:spannningtree}
The Laplacian matrix $\mathcal{L}$ of a directed graph $\mathcal{G}$ has at least one zero eigenvalue with $\emph{ \textbf{1}}_n $ as a corresponding right eigenvector and all nonzero eigenvalues have positive real parts. Furthermore, zero is a simple eigenvalue of $\mathcal{L}$ if and only if $\mathcal{G}$ has a spanning converging tree, i.e., it has at least one vertex accessible from all other vertices.
\end{lemma}

Obviously, the digraph $\mathcal{G}$ depicted in Fig.~\ref{fig:graph} is simple, unweighted, and contains a spanning converging tree.

The system of $n$ agents \eqref{eq:agent0} with feedback protocol~\eqref{eq:controlprotocol1} obeys the following dynamics:
\begin{equation}
\label{eq:syssecondorder}
\ddot x+\gamma \dot x = - \mathcal{L}x,
\end{equation}
where $x = [x_{1}, \ldots, x_{n}]^{\top}$ and $\mathcal{L}$ is the Laplacian matrix associated with the dependency
digraph $\mathcal{G}$:

\begin{equation}
\label{eq:laplacian}
\mathcal{L} =
\begin{bmatrix}[r]
2 & -1 & 0 & 0 & \cdots & 0 & -1\\
-1 & 1 & 0 & 0 & \cdots & 0 & 0 \\
0 & -1 & 2 & -1 & \cdots & 0 & 0  \\
\vdots & \vdots & \ddots & \ddots & \ddots & \vdots & \vdots \\
0 & \cdots & 0 & -1 & 1 & 0 & 0 \\
0 & \cdots & 0 & 0 & -1 & 2 & -1  \\
0 & \cdots & 0 & 0 & 0 & -1 & 1
 \end{bmatrix}.
\end{equation}

Also note that this  closed-loop network system can be equivalently described by
\begin{equation}
\label{eq:network0}
\dot \xi = F\xi,
\end{equation}
where $F\in \mathbb{R}^{2n \times 2n} $ has the form
\begin{equation}
\label{eq:F}
F = I \otimes A -\mathcal{L} \otimes BK,
\end{equation}
and $I \in \mathbb{R}^{n \times n}$ denotes the identity matrix.

Although the consensus conditions \eqref{eq:syssecondorder}, \eqref{eq:network0} can be verified in a straightforward way, this may be computationally expensive in the case of a large number of agents. The framework presented below allows to reduce the problem to a couple of simpler ones using the notion of consensus region. The criterion for scalar agents was first proposed by Polyak and Tsypkin in \cite{PolyakTsypkin1996}; similar results were obtained later in \cite{HaraHayakawaetal2007} and \cite{HaraTanakaIwasaki2014}. Some other extensions may be found in \cite{Lewisetal2014} and \cite{LiDuan2017}.

\subsection{Consensus Region}
Suppose that the consensus problem is studied for a networked dynamical system described by the equation
$$
\phi(s)x = -\mathcal{L}x,
$$
where $\phi(s)$ is a scalar polynomial, $s  :=  \frac{d}{dt}$, and $\mathcal{L}$ is the Laplacian matrix of the dependency digraph containing a spanning converging tree.
\begin{definition}[\hspace{-4pt} \cite{PolyakTsypkin1996, HaraTanakaIwasaki2014, LiDuan2017}]
\label{def:consensusregion}
The $\Omega$-region  of the function $\phi(s)$ is the set of points $\lambda$ on the complex plane for which the function $\phi(s) -\lambda$ has no zeros in the closed right half-plane\emph{:}
$$
\Omega = \{\lambda \in \mathbb{C}: \phi(s)-\lambda \neq 0 \text{ whenever } \emph{Re}(s) \geq 0\}.
$$
\end{definition}
Note that $\Omega$ is precisely the region of parameters $\lambda$ such that the matrix $A-\lambda BK$ is Hurwitz stable. The function $\phi(s)$ is sometimes referred to as the \emph{generalized frequency variable}~\cite{HaraHayakawaetal2007, HaraTanakaIwasaki2014}.

\begin{lemma}[\hspace{-4pt} \cite{PolyakTsypkin1996, HaraTanakaIwasaki2014, LiDuan2017}]
\label{lem:consensusregion}
The system described by \eqref{eq:agent1} reaches consensus under protocol \eqref{eq:controlprotocol1} if and only if
$$
\lambda_i \in \Omega,\quad i \in 2:n,
$$
where $\lambda_i$, $i \in 2:n$, are the nonzero eigenvalues of $-\mathcal{L}$.
\end{lemma}

The details of determining the consensus region may be found in \cite{PolyakTsypkin1996}; in the case of $\phi(s) = s^2+\gamma s$, $\gamma > 0$, this region has form of the interior of a parabola in the complex plane: $\phi(j\omega) = -\omega^2+j\gamma \omega$, $-\infty < \omega < \infty$.


\subsection{Cyclotomic Equation}

\begin{lemma}[\hspace{-4pt} \cite{Neumann2007}]
\label{lem.cyc}
The roots of the cyclotomic equation
$$
z^m -1 = 0
$$
are the de Moivre numbers
$$
z_k = e^{j \frac{2 \pi k}{m}}, \; k \in 0:m-1.
$$

They form a regular polygon with each vertex lying on the unit circle in the complex plane.
\end{lemma}

\subsection{Cassini Ovals}
\begin{definition}[\hspace{-4pt} \cite{Lawrence1972}]
	\label{def:cassini}
The Cassini curve \emph{(}or the Cassini ovals\emph{)} is a quartic curve  defined as the set of points in the plane such that the product of the distances \emph{(}denoted by $b^2$\emph{)} to two fixed points $(a, 0)$ and $(-a, 0)$ is constant\emph{:}
\begin{equation}\label{eq.cascanonical}
[(x-a)^2+y^2][(x+a)^2+y^2] = b^4.
\end{equation}
\end{definition}

\section{Main Results}
Based on the above problem formulation and preliminary results, in this section we derive analytic conditions of achieving consensus for network systems \eqref{eq:syssecondorder}, \eqref{eq:network0}. These conditions turn out to be independent of the number of agents.

First, we study the spectrum of the Laplacian matrix defined by \eqref{eq:laplacian} and obtain the curves which are the loci of its eigenvalues.

\begin{lemma}
\label{lem:laplacian}
The eigenvalues of Laplacian matrix \eqref{eq:laplacian} have the form
$$
\lambda_k^{(1,2)} = 1.5 \pm 0.5\sqrt{5 + 4e^{j \frac{2 \pi k}{m}}}
$$
and lie on the Cassini ovals defined by
\begin{equation}\label{eq:casmy}
[(\tilde x-\sqrt{5})^2+\tilde y^2][(\tilde x+\sqrt{5})^2+\tilde y^2] = 2^4,
\end{equation}
where $\tilde x = 2(x-3/2)$, $\tilde y = 2y$.

\end{lemma}
\begin{proof}
According to Theorem~4 in \cite{AgaevChebotarev2010} our graph is \emph{essentially cyclic} (i.e., the spectrum of $\mathcal{L}$ contains non-real eigenvalues) and the characteristic polynomial of $\mathcal{L}$ is of the form
$$
(Z_2(\lambda))^m -1,
$$
where $Z_n(\lambda) = (\lambda-2)Z_{n-1}(\lambda)-Z_{n-2}(\lambda)$ is a modified Chebyshev polynomial of the second kind with the initial conditions $Z_0(\lambda) \equiv 1$, $Z_1(\lambda) \equiv \lambda - 1$, see \cite{AgaevChebotarev2010} for details.
Armed with the knowledge gained from Theorem~4 in \cite{AgaevChebotarev2010} and Lemma~\ref{lem.cyc} we find   $Z_2 = \lambda^2 -3\lambda +1$, the characteristic polynomial of $\mathcal{L}$ is
\begin{equation}\label{eq.polynomial0}
(\lambda^2 -3\lambda +1)^m -1,
\end{equation}
and its roots are
\begin{equation}\label{eq.polynomial}
\lambda_k^{(1,2)} = 1.5 \pm 0.5\sqrt{5 + 4e^{j \frac{2 \pi k}{m}}}.
\end{equation}

Let the $k$th root of unity $z_k = e^{j \frac{2 \pi k}{m}}$ be
 $a_k+jb_k$, $k \in 0:m-1$. Then
\begin{equation}\label{eq.circle}
a_k^2 +b_k^2 = 1.
\end{equation}

The $2m$ roots of \eqref{eq.polynomial0} can be found from the equation
$$
\lambda^2 -3\lambda +1 - a_k - jb_k = 0, \; k = 0,1,\ldots, m-1.
$$
Let $\lambda = x+jy$. Then from
$$
x+jy = 1.5 \pm 0.5\sqrt{5 + 4a_k+j4b_k}
$$
we find $a_k = (x-3/2)^2-y^2-5/4$ and $b_k = 2yx - 3y$.
Taking into account \eqref{eq.circle} one arrives at
$$
((x-3/2)^2-y^2-5/4)^2+4y^2(x-3/2)^2=1.
$$
The last equation can be simply rewritten in the form of \eqref{eq:casmy}.
\end{proof}

The eigenvalues of $\mathcal{L}$ with $n=16$ and $n = 44$ are shown in Fig.~\ref{fig:eigenvalues}. The egg-shaped Cassini ovals can be recognized easily.
\begin{figure}[h!]
  \centering
    \includegraphics[width=9.5cm]{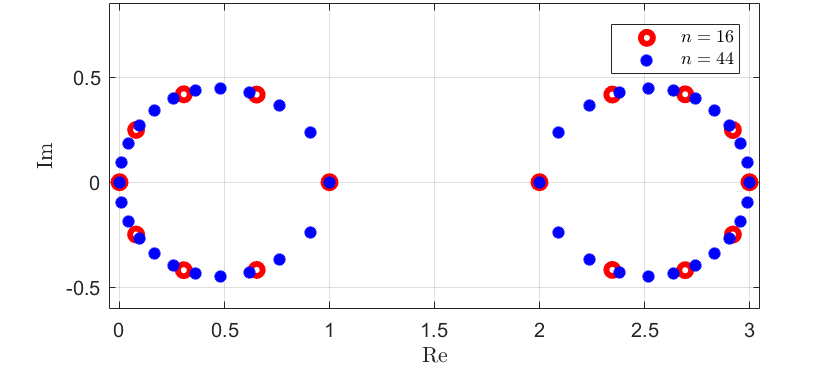}
    \caption{The eigenvalues of $\mathcal{L}$ with $n=16$ and $n = 44$.}
  \label{fig:eigenvalues}
\end{figure}

\begin{corollary}
\label{cor:eigandcurve}
The eigenvalues of $-\mathcal{L}$ are
\begin{equation}
\label{eq.eigenvaluessys}
\lambda_k^{(1,2)} = -1.5 \pm 0.5\sqrt{5 + 4e^{j \frac{2 \pi k}{m}}}
\end{equation}
and the equation of the corresponding Cassini curve is
\begin{equation}
\label{eq.cassini}
((x+3/2)^2-y^2-5/4)^2+4y^2(x+3/2)^2-1=0.
\end{equation}
\end{corollary}
\begin{remark}
\label{rem:cocnsensuscrit}
According to consensus criterion \emph{\cite{PolyakTsypkin1996, HaraTanakaIwasaki2014, LiDuan2017}} for the second-order agents \eqref{eq:agent0}, the consensus is reached if and only if all the nonzero eigenvalues of $-\mathcal{L}$ lie in the $\Omega$-region of $\phi(s) = s^2+\gamma s$. Checking this criterion requires computation of the spectrum. Below we propose a condition independent of the number of agents. It is sufficient yet simple and not that much conservative.
\end{remark}

\begin{theorem}
\label{th:consensuscondition}
The system of interconnected second-order agents \eqref{eq:network0} reaches a consensus in the sense of \eqref{eq:consensus} for all $\gamma > \sqrt{\frac{6}{7}}$.
\end{theorem}
\begin{proof}
Since $\phi(j\omega) = -\omega^2+j\gamma \omega$, the equation can be rewritten as $y^2 = -\gamma^2 x$. Therefore, the consensus condition can be reduced to the condition that the Cassini ovals described by \eqref{eq.cassini} in Corollary~\ref{cor:eigandcurve} belong to the interior of the parabola $ y^2 = -\gamma^2x$ without intersection (except for the one at the origin), see Fig.~\ref{fig:omegacassini_1}.
\begin{figure}[h!]
  \centering
    \includegraphics[width=9cm]{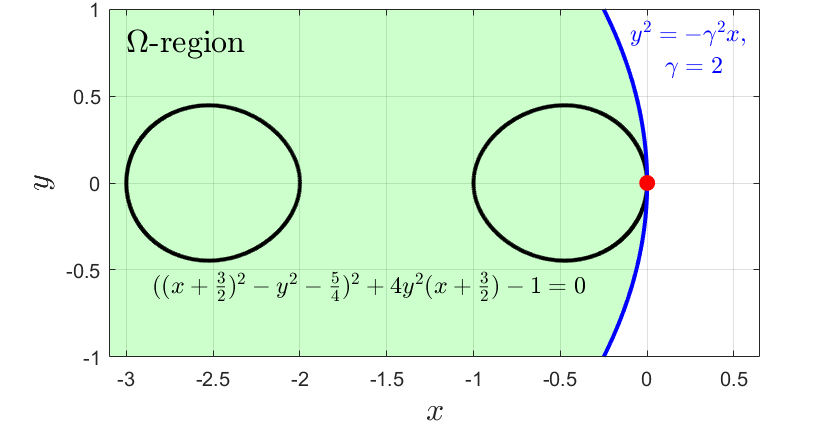}
    \caption{$\Omega$-region bounded by $y^2 = -\gamma^2 x$ and the Cassini ovals, $\gamma = 2$.}
  \label{fig:omegacassini_1}
\end{figure}
Substituting this into \eqref{eq.cassini} we arrive at the equation in $x$ with parameter $\gamma$:
\begin{equation}
\label{eq:paramequation}
x(x^3+(6-2\gamma^2)x^2+(\gamma^4-6\gamma^2+11)x+6-7\gamma^2) = 0.
\end{equation}

The zero root is out of our interest, thus let us study the properties of the cubic polynomial
\begin{equation}
\label{eq:paramequationcubic}
x^3+(6-2p)x^2+(p^2-6p+11)x+6-7p = 0,
\end{equation}
where $p = \gamma^2$.
The case  that determines the required condition corresponds to one intersection, and therefore one real negative root and a pair of roots with the same sign of their real parts. According to Vieta's formulas, a product of the roots satisfies $x_1x_2x_3 = 7p-6$ and is negative. Therefore, $p>6/7$ implies the absence of any intersection except for the one at (0,~0).

The consensus condition independent of the number of agents follows immediately.

\begin{figure}[h!]
  \centering
    \includegraphics[width=9cm]{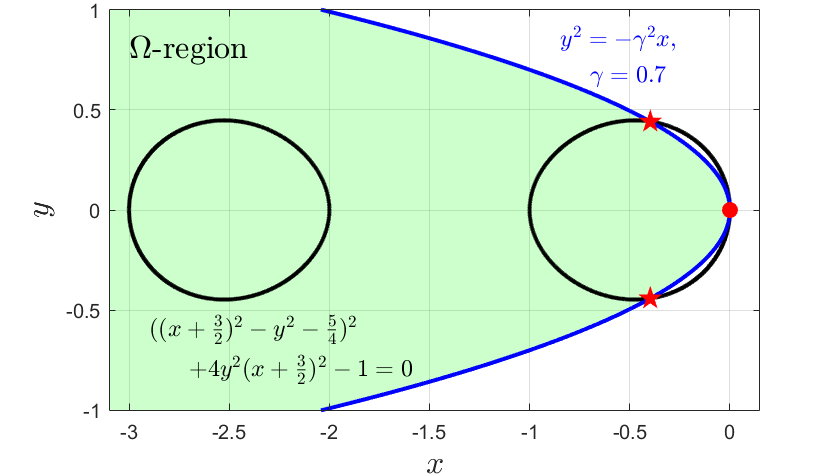}
    \caption{$\Omega$-region bounded by $y^2 = -\gamma^2 x$ and the Cassini ovals, $\gamma = 0.7$.}
  \label{fig:omegacassini_2}
\end{figure}
\end{proof}

\begin{remark}
In this paper, we limit ourselves to the case of $m\geq3$ due to the fact that for $m=1$, the considered graph does not exist and the case of $m=2$ results in a real spectrum of $\mathcal{L}$~\emph{\cite{AgaevChebotarev2010}} which implies reaching consensus for any $\gamma > 0$.
\end{remark}

\begin{remark}
By virtue of the consensus region approach, a similar condition for second-order agents whose dependency digraph is a Hamiltonian cycle \emph{(}Fig.~\emph{\ref{fig:graphcyc}}\emph{)} can be derived. In this case, the corresponding Laplacian matrix is a circulant matrix \emph{(}see, e.g., \emph{\cite{Bernstein2009}} for details\emph{)}. Its eigenvalues are located on a unit circle centered at $(-1,\; j0)$.
\begin{figure}[h!]
  \centering
    \includegraphics[width=5.5cm]{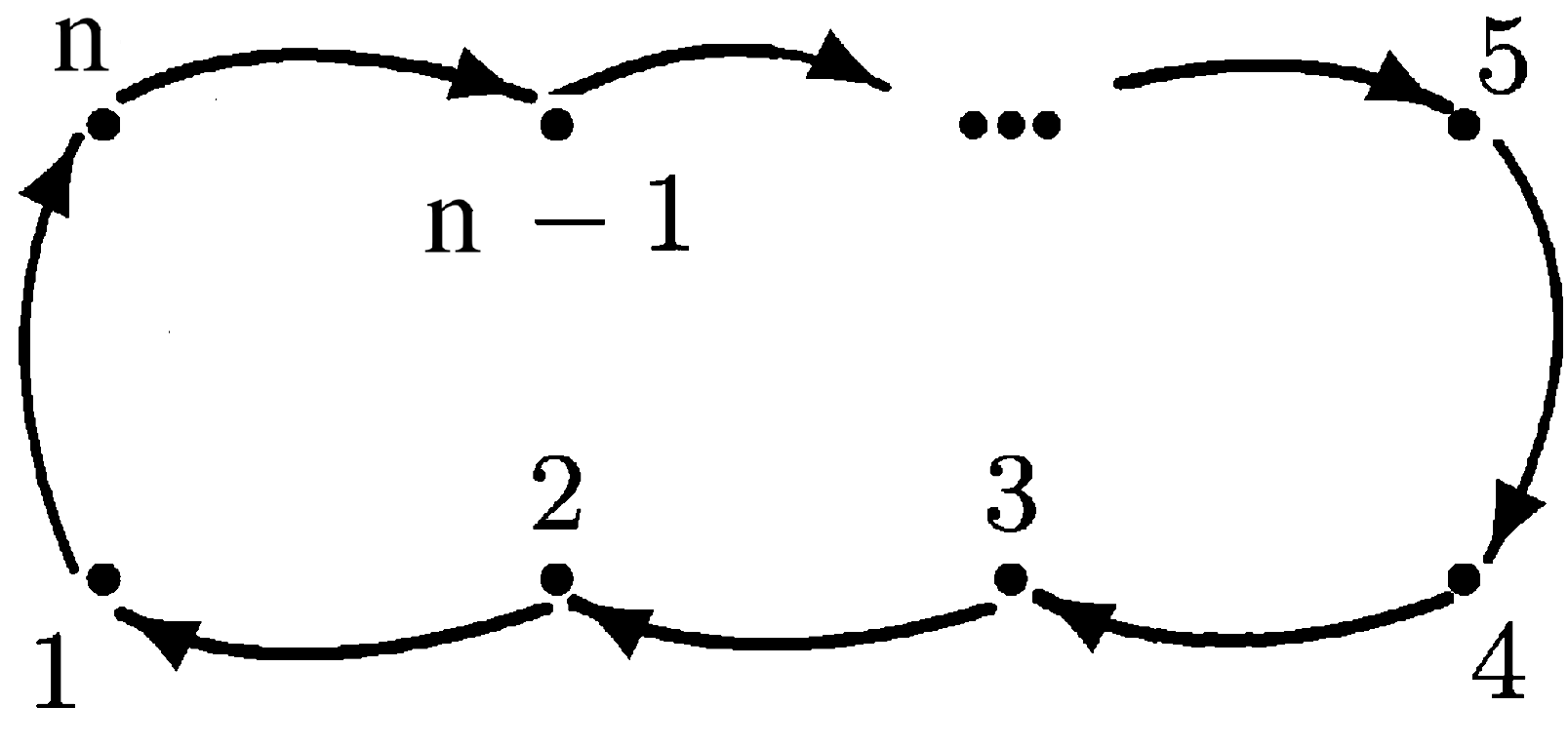}
    \caption{Diraph of the cyclic pursuit.}
  \label{fig:graphcyc}
\end{figure}
\begin{figure}[h!]
  \centering
    \includegraphics[width=9cm]{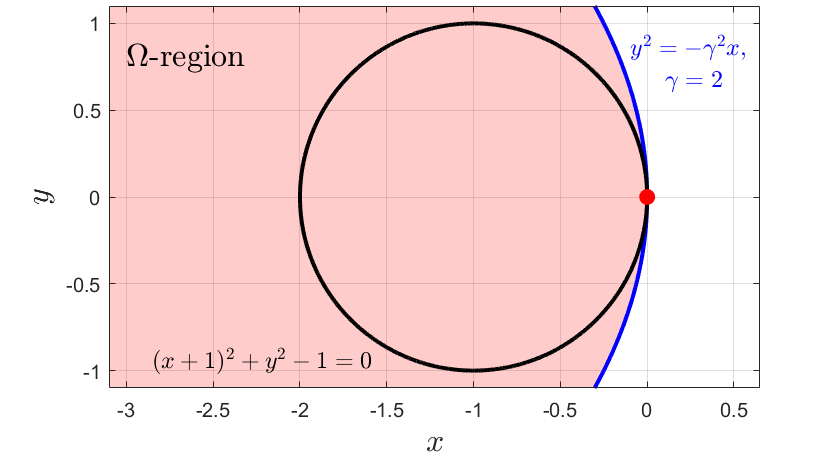}
    \caption{$\Omega$-region bounded by $y^2 = -\gamma^2 x$ and the unit circle, $\gamma = 2$.}
  \label{fig:omegacircle1}
\end{figure}
Recall that the consensus region defined by the transfer function of agent \eqref{eq:agent0} is the interior of a parabola. It can be verified that the equation for the intersection writes as $x^2+2x-\gamma^2x=0$, which implies that the absence of intersection \emph{(}except for the one at the origin\emph{)} holds true for $\gamma > \sqrt{2}$. This demonstrates that the consensus condition for second-order agents \eqref{eq:agent0} in cyclic pursuit can be derived easier as compared to, e.g., \emph{\cite{Sharmaetal2013}}.

The consensus condition for the case of undirected (or, in terms of digraph, two-directional) ring topology is even simpler. Since the spectrum of the Laplacian matrix is completely real, the eigenvalues $\lambda_2, \ldots, \lambda_n$ all belong to the $\Omega$-region for any $\gamma >0$ providing consensus in the sense of \eqref{eq:consensus}.
\end{remark}


\section{Numerical Examples}
First, we demonstrate the low conservatism of the proposed criterion. 

The sufficient condition of Theorem~\ref{th:consensuscondition} guarantees convergence to consensus in the sense of \eqref{eq:consensus} whenever $\gamma > \sqrt{6/7}$. However, small groups of agents may still reach consensus with dampings $\gamma \leq \sqrt{6/7} $: if all the nonzero eigenvalues $\lambda_i$, $i \in 2:n$, of $-\mathcal{L} $  lie inside the $\Omega$-region, then the conditions of the general criterion formulated in Lemma~\ref{lem:consensusregion} are satisfied (which implies that the corresponding eigenvalues $\hat \lambda_i$, $i \in 2:2n$, of the system matrix $F$ represented by \eqref{eq:F}  lie in the open left half-plane of the complex plane). This can be verified as follows. Let us set $\gamma = \sqrt{6/7}$ and compute the eigenvalues of $F$ for various numbers of agents $n=6,8,10, \ldots$. Taking the maximum of their real parts we observe that this maximum approaches zero as the number of agents increases, see Fig.~\ref{fig_pic4}.
\vspace{-2ex}
 \begin{figure}[h]
	\centerline{\includegraphics[width=9.5cm]{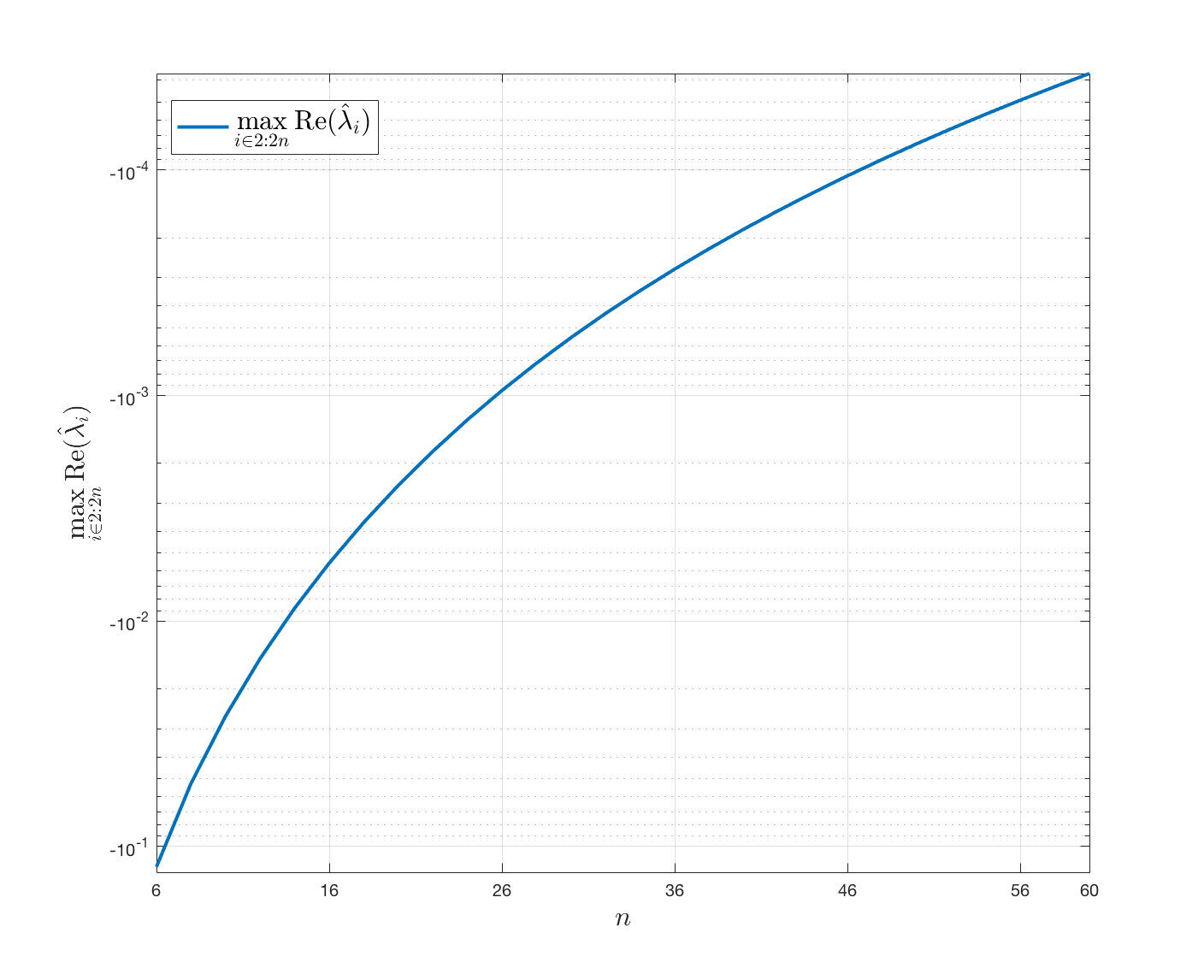}}
	\vspace{-2ex}
	\caption{The values of $\max\limits_{i \in 2:2n}{\rm{Re}}(\hat \lambda_i) $ subject to $n = 6,8,10,\ldots, 60$ (logarithmic scale); the damping coefficient is $\gamma = \sqrt{6/7}$.}
	\label{fig_pic4}
\end{figure}

Now, to illustrate the dynamics of convergence to consensus, consider a formation of $n=50$ identical second-order agents
\eqref{eq:agent0} linked by the topology depicted in Fig.~\ref{fig:graph}. Let the initial state of each agent be randomly generated within $[0, 10]$.

System \eqref{eq:network0} exhibits stable behavior of the transients with a damping coefficient $\gamma = 2$, see Fig.~\ref{fig_pic1}. Computation of  the eigenvalues of $F$ yields
$\max\limits_{i \in 2:2n}{\rm{Re}}( \hat \lambda_i) = -0.0032.$
 \begin{figure}[h]
	\centerline{\includegraphics[width=9.5cm]{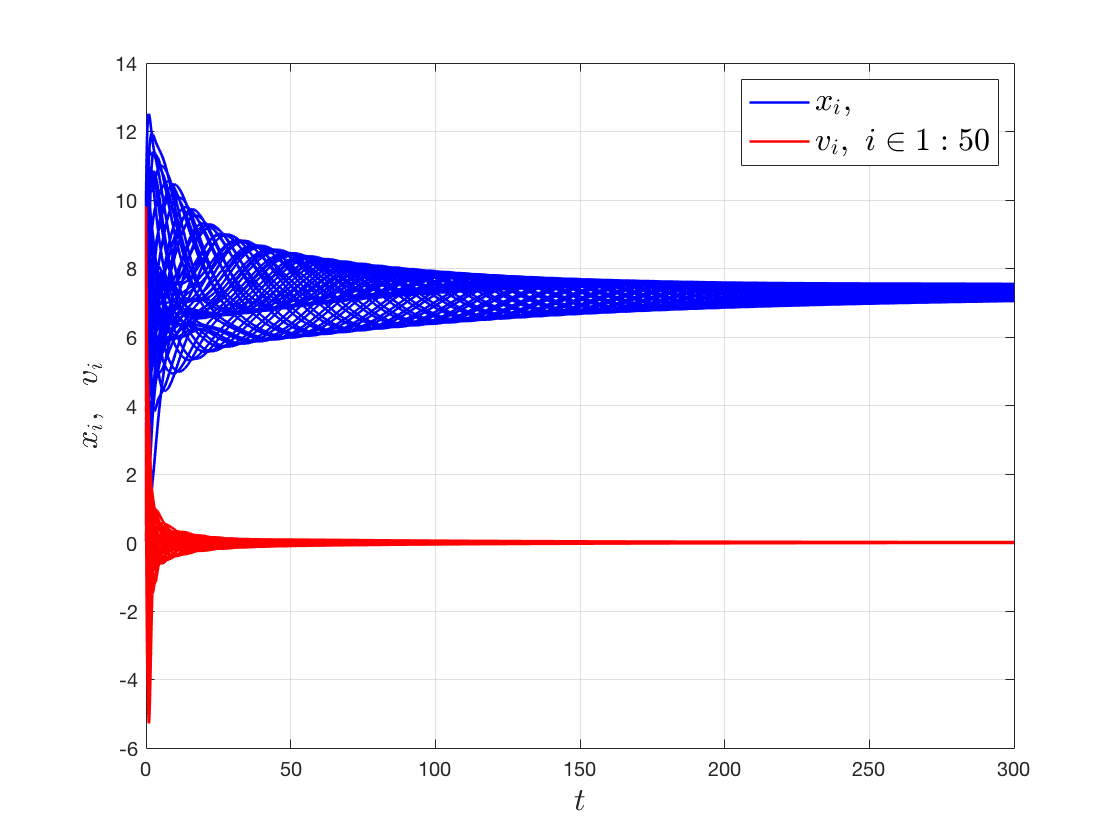}}
	\caption{The trajectories of the system \eqref{eq:syssecondorder} with $n=50$ agents, $\gamma = 2$.}
	\label{fig_pic1}
\end{figure}

In the case of smaller gain $\gamma = 0.95$ close enough to the stability margin $\gamma = \sqrt{6/7} \approx 0.926$ presented by Theorem~\ref{th:consensuscondition} we still observe similar behavior, see Fig.~\ref{fig_pic2}.
 \begin{figure}[h]
	\centerline{\includegraphics[width=9.5cm]{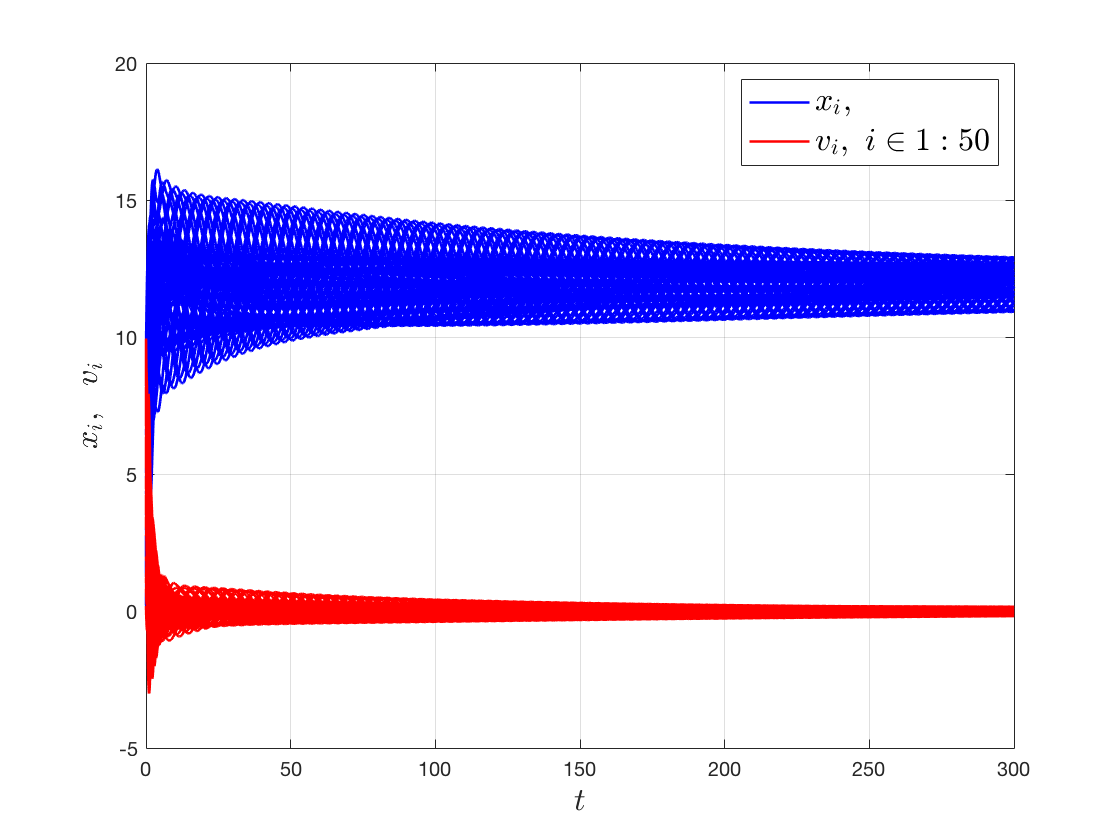}}
	\caption{The trajectories of the system \eqref{eq:syssecondorder} with $n=50$ agents, $\gamma = 0.95$.}
	\label{fig_pic2}
\end{figure}

The dynamics becomes different when the consensus condition is violated. For example, for $\gamma = 0.9$ the transients diverge; $\max\limits_{i \in 2:2n}{\rm{Re}}(\hat \lambda_i) = 0.000633$.
 \begin{figure}[h]
	\centerline{\includegraphics[width=9.5cm]{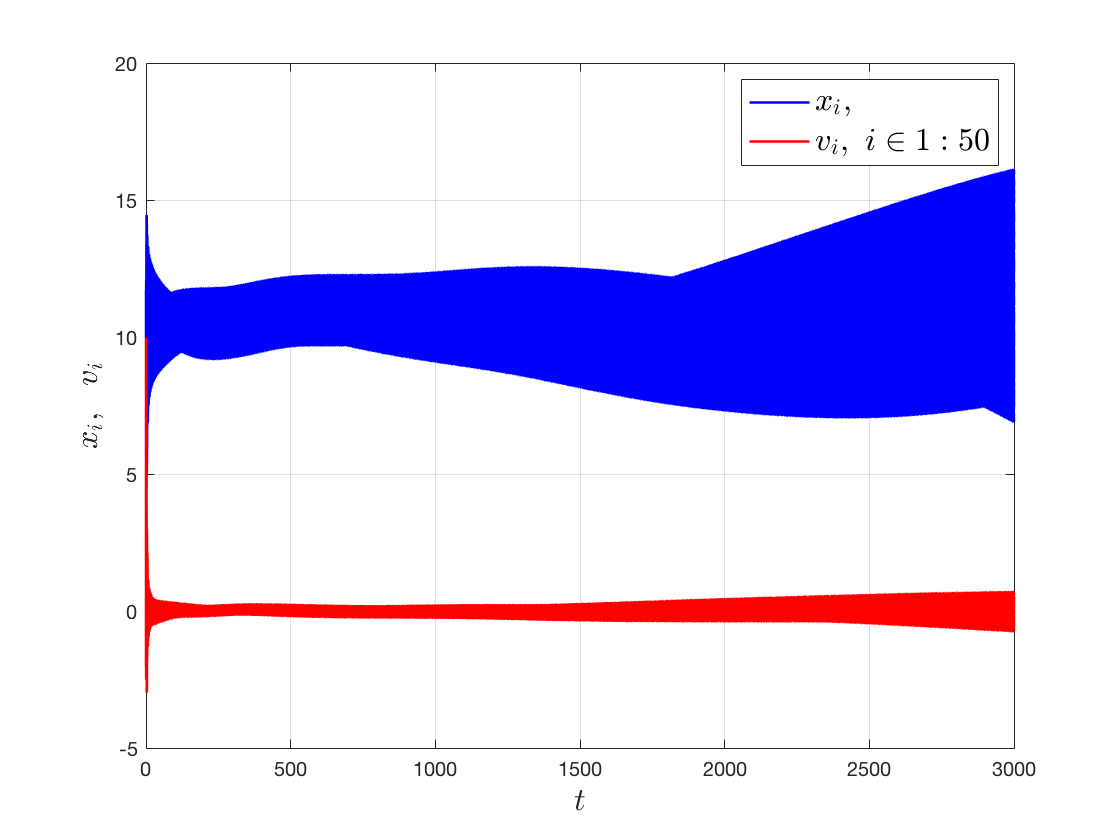}}
	\caption{The trajectories of the system \eqref{eq:syssecondorder} with $n=50$ agents, $\gamma = 0.9$.}
	\label{fig_pic3}
\end{figure}
Fig.~\ref{fig_pic3} presents the results of simulations for this unstable case.

Table~1 illustrates the dependence of the actual consensus margin $\gamma$ on the number of agents $n$.
\begin{table}[ht]
	\caption{Consensus margins}
	\centering 
	\begin{tabular}{| l | l | l | l | l | l | l |}
		\hline
		$n$ & 10 & 20 & 30 & 40 & 50 & 60 \\ \hline
		$\gamma$ & 0.8195 & 0.8999 & 0.9149 & 0.9195 & 0.9218 & 0.9230 \\
		\hline
	\end{tabular}
\end{table}


Therefore, the simulation results confirm the theoretical conclusions of Theorem~\ref{th:consensuscondition} and demonstrate low conservatism of the provided sufficient condition of consensus for a large number of agents.

\section{Conclusions}
The contribution of the paper is threefold:

\begin{itemize}
	\item First, the specific communication ring topology was investigated; it was discovered that the eigenvalues of the corresponding Laplacian matrix lie on the Cassini ovals;
	\item Second,  the criterion for reaching consensus was proposed that relies solely on the value of the damping coefficient of a single agent and neglects the number of them;
	\item Third, the theoretical results were supplemented by numerical experiments. It was shown that for a large number of interacting agents, the sufficient condition presented by Theorem~\ref{th:consensuscondition} is quite close to the necessary one.
\end{itemize}

The possible extensions include establishing consensus conditions for second-order agents distributed on other regular topologies together with the control protocols that use solely local measurements. These problems will be the subject of continuing research.


\begin{thebibliography}{00}


\bibitem{ProskurnikovTempo2017}
A. V. Proskurnikov and R. Tempo, ``A Tutorial on Modeling and Analysis of Dynamic Social Networks. Part I,'' Annual Reviews in Control, no. 43, pp. 65--79, 2017.


\bibitem{Bullo2018}
F. Bullo, Lectures on Network Systems (With contributions by J. Cort\'es, F. D\"orfler, and S. Mart\'inez), 2018, http://motion.me.ucsb.edu/book-lns/


\bibitem{Lewisetal2014}
F. L. Lewis, H. Zhang, K. Hengster-Movric, and A. Das, Cooperative Control of Multi-Agent Systems:
Optimal and Adaptive Design Approaches, London: Springer, 2014.
	

\bibitem{LiDuan2017}
Z. Li and Z. Duan, Cooperative Control of Multi-Agent Systems: A Consensus Region Approach, CRC Press, 2017.


\bibitem{Darboux1878}  J.~G.~Darboux,
\newblock ``Sur un Probl\`eme de G\'eom\'etrie \'El\'ementaire,'' \newblock Bulletin des Sciences Math\'ematiques et Astronomiques, vol. 2, no. 1, pp. 298--304, 1878.

\bibitem{Nahin2007} P.~J.~Nahin,
\newblock Chases and Escapes: The Mathematics of Pursuit and Evasion,
\newblock Princeton University Press, 2007.

\bibitem{Sharmaetal2013}
B. R. Sharma, S. Ramakrishnan, and M. Kumar, ``Cyclic Pursuit in a Multi-Agent Robotic System with Double-Integrator Dynamics under Linear Interactions,'' Robotica, vol. 31, no. 7, pp. 1037--1050, 2013.


\bibitem{ElorBruckstein2011}
Y. Elor and A. M. Bruckstein, ``Uniform Multi-Agent Deployment on a Ring,'' Theor. Comput. Sci., vol. 412, no. 8--10, pp. 783--795, 2011.


\bibitem{MarshallBrouckeFrancis2004} J.~A.~Marshall, M.~E.~Broucke and B.~A.~Francis,
\newblock ``Formations of Vehicles in Cyclic Pursuit,''
\newblock IEEE Transactions On Automatic Control, vol. 49, no. 11, pp. 1963--1974, 2004.


\bibitem{SmithBrouckeFrancis2005} S.~L.~Smith, M.~E.~Broucke, and B.~A.~Francis,
\newblock ``A Hierarchical Cyclic Pursuit Scheme for Vehicle Networks,''
\newblock Automatica, vol. 41, no. 6, pp. 1045--1053, 2005.


\bibitem{MukherjeeGhose2016}
D. Mukherjee and D. Ghose,``Generalized Hierarchical Cyclic Pursuit,'' Automatica,
vol. 71, pp. 318--323, 2016.


\bibitem{SinhaGhose2006}
A. Sinha and D. Ghose, ``Generalization of Linear Cyclic Pursuit with Application to Rendezvous of Multiple Autonomous Agents,'' IEEE Transactions on Automatic
Control, vol. 51, no. 11, pp. 1819--1824, 2006.


\bibitem{ElmachtoubVanLoan2010}
A. N. Elmachtoub and C. F. van Loan, ``From Random Polygon to Ellipse. An Eigenanalysis,'' SIAM Rev., vol. 52, no. 1, pp. 151--170, 2010.


\bibitem{Shcherbakov2011}
P. S. Shcherbakov, ``Formation Control: The Van Loan Scheme and Other Algorithms,'' Autom. Remote Control, vol. 72, no. 10, pp. 2210--2219, 2011.


\bibitem{RamirezRiberosPavoneetal2010}
J. L. Ramirez-Riberos, M. Pavone, E. Frazzoli, and D. W. Miller, ``Distributed Control of Spacecraft Formations via Cyclic Pursuit. Theory and Experiments,'' AIAA J. Guidance, Control, Dynamics, vol. 33, no. 5, pp. 1655--1669, 2010.


\bibitem{MukherjeeZelazo2018}
D. Mukherjee and D. Zelazo, ``Robust Consensus of Higher Order Agents over Cycle Graphs,'' in Proc.  of the 58th Israel Annual Conference on Aerospace Sciences, pp. 1072--1083, 2018.


\bibitem{WagnerBruckstein1997}
I. A. Wagner and A. M. Bruckstein, ``Row Straightening via Local Interactions,'' Circuits Syst. Signal Process., vol. 16, no. 2, pp. 287--305, 1997.


\bibitem{KvintoParsegov2012}
Ya. I. Kvinto and S. E. Parsegov, ``Equidistant Arrangement of Agents on Line: Analysis of the Algorithm and Its Generalization,'' Autom. Remote Control, vol. 73, no. 11, pp. 1784--1793, 2012.


\bibitem{ProskurnikovParsegov2016}
A. V. Proskurnikov and S. E. Parsegov, ``Problem of Uniform Deployment on a Line Segment for Second-Order Agents,'' Autom. Remote Control, vol. 77, no. 7, pp. 1248--1258, 2017.


\bibitem{Goldin2013}
D.~Goldin, ``Double Integrator Consensus Systems with Application to Power Systems,'' in Proc. 4th IFAC Workshop NecSys-2013, pp. 206--211, 2013.



\bibitem{RenCao2011}
W. Ren and Y. C. Cao, Distributed Coordination of Multi-Agent Networks, Springer, London, 2011.


\bibitem{Ren2008a}
W. Ren. ``On Consensus Algorithms for Double-Integrator Dynamics,'' IEEE Trans. Autom. Control,
vol.~53, no.~6, pp.~1503--1509, 2008.


\bibitem{AgaevChebotarev2002}
P. Chebotarev and R. Agaev, ``Forest Matrices around the Laplacian Matrix,''
Linear Algebra and Its Applications, vol. 356, pp. 253--274, 2002.


\bibitem{AgaevChebotarev2000}
R. P. Agaev and P. Y. Chebotarev, ``The Matrix of Maximum Out Forests of a Digraph and Its Applications,'' Autom. Remote Control, vol. 61, no. 9, pp. 1424--1450, 2000.

\bibitem{AgaevChebotarev2001}
R. P. Agaev and P. Yu. Chebotarev, ``Spanning Forests of a Digraph and Their Applications,'' Autom. Remote Control, vol. 62, no. 3, pp. 443--466, 2001.


\bibitem{PolyakTsypkin1996}
B.~T. Polyak and Y.~Z. Tsypkin,``Stability and Robust Stability of Uniform Systems,'' Autom. Remote Control, vol. 57, no. 11, pp. 1606--1617, 1996.

\bibitem{HaraTanakaIwasaki2014}
S. Hara, H. Tanaka, and T. Iwasaki, ``Stability Analysis of Systems with Generalized Frequency Variables,'' IEEE Trans. Autom. Control, vol.~59, no.~2, pp.~313--326, 2014.


\bibitem{HaraHayakawaetal2007}
S. Hara, T. Hayakawa, and H. Sugata, ``Stability Analysis of Linear Systems with Generalized Frequency Variables and Its Applications to Formation Control,'' in Proc. IEEE Conf. Decision Control, pp. 1459--1466, 2007.


\bibitem{Neumann2007}
O. Neumann, ``Cyclotomy: From Euler through Vandermonde to Gauss,'' in  Leonhard Euler: Life, Work and Legacy, R.~E.~Bradley and C.~E.~Sandifer, Eds. Amsterdam: Elsevier, 2007, pp. 323--362.


\bibitem{Lawrence1972}
J. D. Lawrence, A Catalog of Special Plane Curves, New York: Dover,
1972.


\bibitem{AgaevChebotarev2010}
R.~P. Agaev and P.~Yu. Chebotarev, ``Which Digraphs with Ring Structure are Essentially Cyclic?'' Adv. Appl. Math. vol. 45, pp. 232--251, 2010.



\bibitem{Bernstein2009}
D. Bernstein, Matrix Mathematics: Theory, Facts, and Formulas,  Princeton University Press, 2009.














	





























\end{thebibliography}
\end{document}